%% file: wabi.tex
\documentclass{llncs}

\usepackage{verbatim}
\usepackage{wrapfig}
\usepackage{amsmath}
\usepackage{algorithmic}
\usepackage{algorithm}
\usepackage{graphicx}
\usepackage{caption,subcaption}

\renewcommand{\algorithmicrequire}{\textbf{Input:}}
\def\argmax{\mathop{\rm arg\,max}}

\title{Faster Algorithms for RNA-folding using the Four-Russians method}
\author{Balaji Venkatachalam \and Dan Gusfield \and Yelena Frid}
\institute{Department of Computer Science, UC Davis\\
\{balaji, gusfield, yafrid\}@cs.ucdavis.edu}

\begin{document}
\maketitle

\begin{abstract}
\input{abs}

\end{abstract}

\input{rna}

   \bibliographystyle{abbrv} 
   \bibliography{thesis}
\end{document}

%% file: abs.tex
The secondary structure that maximizes the number of
non-crossing matchings between complimentary bases of an RNA sequence
of length $n$ can be computed in $O(n^3)$ time using Nussinov's dynamic
programming algorithm.  The Four-Russians method is a technique that will
reduce the running time for certain dynamic programming algorithms by
a multiplicative factor after a preprocessing step where solutions to all smaller
subproblems of a fixed size are exhaustively enumerated and solved.  Frid and
Gusfield designed an $O(\frac{n^3}{\log n})$ algorithm for RNA folding
using the Four-Russians technique.  In their algorithm the preprocessing is interleaved with the algorithm
computation.

We simplify the algorithm and the analysis by doing the preprocessing
once prior to the algorithm computation. We call this the {\em two-vector}
method.  We also show variants where instead of exhaustive preprocessing,
we only solve the subproblems encountered in the main algorithm once and
memoize the results.  We give a simple proof of correctness and explore
the practical advantages over the earlier method.

The Nussinov algorithm admits an $O(n^2)$ time parallel algorithm. We show
a parallel algorithm using the two-vector idea that improves the time
bound to $O(\frac{n^2}{\log n})$.

We have implemented the parallel algorithm on graphics processing units
using the CUDA platform.  We discuss the organization of the data structures
to exploit coalesced memory access for fast running times.  The ideas
to organize the data structures 
also help in improving the running time of the serial algorithms.
For sequences of length up to 6000 bases the parallel algorithm takes
only about 2.5 seconds and the two-vector serial method takes about 57
seconds on a desktop and 15 seconds on a server. Among the serial algorithms, the two-vector and memoized
versions are faster than the Frid-Gusfield algorithm by a factor of 3,
and are faster than Nussinov by up to a factor of 20.

%% file: rna.tex

\section{Introduction}

Computational approaches to find the secondary structure of RNA molecules
are used extensively in bioinformatics applications.
The classic dynamic programming (DP) algorithm proposed in the 1970s
has been central to most structure prediction algorithms.  While the
objective of the original algorithm was to maximize the number of
non-crossing
pairings between complementary bases, the dynamic programming approach
has been used for other models and approaches, including minimizing
the free energy of a structure.  The DP algorithm runs in cubic time
and there have been many attempts at improving its running time. Here,
we use the Four-Russians method for speeding up the computation.

The Four-Russians method, named after Aralazarov et al.~\cite{four-rus}, is a
method to speed up certain dynamic programming algorithms.  In a typical
Four-Russians algorithm there is a preprocessing step that exhaustively
enumerates and solves a set of subproblems and the results are tabled.  In the main
DP algorithm, instead of filling out or inspecting individual cells,
the algorithm takes longer strides in the table.  The computation for
multiple cells is solved in constant time by utilizing the preprocessed
solutions to the subproblems.  The longer strides to fill the table
reduce the runtime by a multiplicative factor.  The size of the subproblems is chosen
in a way that does not make the preprocessing too expensive.

Frid and Gusfield~\cite{FG} showed the application of the
Four-Russians approach for RNA
folding.  In their algorithm, the preprocessing is
interleaved with the algorithm computation. They fill out a part of
the DP table and use these entries to complete a part of the
preprocessing.  The preprocessed entries are 
used later in the computation.

We show a simpler algorithm where all the preprocessing is completed before the start of the
main algorithm.  This simplifies the correctness proof and the runtime
analysis. 
This approach helps in obtaining a $\log n$ factor improvement for the parallel algorithm. 
In comparing various methods for RNA folding, Zakov and Frid (personal
communication)
had independently observed that the algorithm in~\cite{FG} could be
modified to do the preprocessing at once.  It is essentially the idea
as described here.

In this paper we explore the implications of the one-pass preprocessing
idea.  This description of the algorithm leads naturally to two other
variants. We empirically evaluate these variants and also the
implementation of the parallel algorithm.  


The parallel architecture of general-purpose graphical processing
units (GPUs) have  been exploited for many real-world application
in addition to applications in gaming and visualization problems. 
GPUs have also been used to speed up RNA folding algorithms~\cite{Chang2010,rizk,SGM2011}.
Here we show how the Four-Russians method
allows an organization of the data structures for fast memory accesses.
We also describe the organization of the parallel hierarchy to exploit
the inherent parallelism of the solution.

In the rest of the section, we describe the problem in relation to the
other problems in RNA folding.  To keep the paper self-contained, we will
first describe the {\em two-vector
algorithm}, our application of the Four-Russians method  to the RNA folding problem.  We will use that description to
describe the original Four-Russians method for RNA folding by Frid and
Gusfield~\cite{FG}.  This discussion leads to two other variants where the preprocessing
is done on demand, instead of the exhaustive preprocessing in the
two-vector method and the Frid-Gusfield algorithm.  In section~\ref{sec:par} we
discuss the $O(n^2/ \log n)$ parallel algorithm.  We will then describe
the implementation of a parallel algorithm using CUDA.  The final sections
have discussion on empirical observations and conclusions.  Due to
space limitations, this manuscript focuses mostly on the theoretical
aspects and describes the experimental results briefly.  Detailed
discussion can be found in~\cite{RNA-tr}.

\paragraph{Related work}

The $O(n^3)$ dynamic programming algorithm due to Nussinov et
al.~\cite{Nussinov1978,Nussinov1980} maximizes the number of
non-crossing matching
complimentary bases. There have been many methods since Zuker and
Stiegler~\cite{Zuker1981} that infer the folding using thermodynamic
parameters~\cite{Tinoco1973,Mathews2004b}
which are more realistic than maximizing the number of base
pairs.  These methods have been implemented in many packages
including UNAFold~\cite{Markham2008}, Mfold~\cite{Zuker2003},
Vienna RNA Package~\cite{Hofacker2003},
RNAstructure~\cite{Reuter2010}.

Probabilistic methods include stochastic context-free
grammars~\cite{Durbin1998,Dowell2004}, the maximum expected accuracy~(MEA)
method, where secondary structures are composed of pairs that have a
maximal sum of pairing probabilities, eg., MaxExpect~\cite{Lu2009},
Pfold~\cite{Knudsen2003}, CONTRAfold~\cite{Do2006} which maximize
the posterior probabilities of base pairs; and Sfold~\cite{Ding2003},
CentroidFold~\cite{Hamada2009} that maximize the centroid estimator.
There are also other methods that use a combination of thermodynamic
and statistical parameters~\cite{Andronescu2007} and methods that
use training sets of known folds to determine their parameters,
eg., CONTRAfold \cite{Do2006}, and Simfold\cite{Andronescu2010} and
ContextFold\cite{Zakov2011}.

In addition to the Four-Russians method, other methods to
improve the running time include Valiant's max-plus matrix
multiplication by Akutsu~\cite{Akutsu1999} and Zakov et al.~\cite{Zakov2010};
and sparsification, where the branch points are pruned to get an improved
time bound~\cite{Wexler2007,Backofen2011Zakov}.

CUDA, the programming platform for GPGPUs,  has been used to solve many
bioinformatics problems. Chang, Kimmer and Ouyang~\cite{Chang2010}
and Stojanovski, Gjorgjevikj and Madjarov~\cite{SGM2011} show an
implementation of the Nussinov algorithm on CUDA.  Rizk et al.~\cite{rizk}
describe the implementation for Zuker and Stiegler method involving
energy parameters.  
These methods are discussed in section~\ref{sec:cuda-rel}.

\section{The Nussinov Algorithm}

In this paper, we consider the basic RNA folding problem of maximizing
the number of non-crossing complimentary base pair matchings.
Complimentary bases can be paired, i.e., {\tt A} with {\tt U} and {\tt
C} with {\tt G}. A set
of disjoint pairs is a matching. The pairs in a matching must not
cross, i.e., if bases in positions $i$ and $j$ are paired and if bases
$k$ and $l$ are paired, then either they are nested, i.e., $i < k < l<
j$ or they are non-intersecting, i.e.,  $i<j<k<l$.  The objective is
to maximize the number of pairings under these constraints.

The following algorithm, due to Nussinov~\cite{Nussinov1978} maximizes the number
of non-crossing matchings.  For an input sequence $S$ of length $n$
over the alphabet {\tt A, C, G, U},
the recurrence is defined as follows. Let $D(i,j)$ denote the optimal
cost of folding for the subsequence from $i$ to $j$.  For all $i$,
$D(i, i-1) = D(i,i) = 0$
and for all $i < j$:
\begin{equation}
D(i,j)  =  \max \left\{
\begin{array}{l}
b(S(i), S(j)) + D(i+1, j-1) \\
\max_{i+1 \leq k \leq j} D(i, k-1)+D(k,j)
\end{array}
\right.
\label{eqn:recur}
\end{equation}
where $b(.,.) = 1$ for complimentary bases and $0$ otherwise.
The DP table is the upper triangular part of the $n\times n$ matrix.
The optimal solution is given by $D(1, n)$.  The table can be filled
column-wise from the first column till the $n^\textrm{th}$.  There are other ways of
filling the table too, eg., along the diagonals --- the $(i,i)$-diagonal
first, $(i, i+1)$-diagonal next and so on, until the last diagonal with
one entry, $D(1,n)$.  To allow for traceback we need to store the bases
that are paired to get the maximum value.  Let $D^*(i,j)$ denote the
corresponding indices.  These are obtained by substituting $\argmax$
in place of $\max$ in the above recurrence and can be computed along
with the $\max$ value.

The first part of the recurrence can be solved in constant time. The
second part is more expensive, incurring $\Theta(n)$ look ups and maximum
computations. There are $O(n^2)$ entries in the DP table and each cell
can be computed in $O(n)$ time, giving an $O(n^3)$ time algorithm.

\section{The Four-Russians Algorithms}
In this section we discuss three variants of the Four-Russians
algorithm.  We will first describe the {\em two-vector} approach.
Since it is simpler than the other methods  we will use the
description to discuss two other variants.

\subsection{Two-vector algorithm}
\label{sec:two-vec}
To apply the Four-Russians technique we start with the following
observation:

\begin{lemma} \label{lem:diff}
The values along a column from bottom to top and along a row from left
to right are monotonically non-decreasing.  Consecutive cells differ at
most by 1.  \end{lemma}

\begin{proof}
Consider neighboring cells $(i,j)$ and $(i+1,j)$.  $D(i,j)$ represents
the solution of a longer sequence than $D(i+1,j)$. Therefore the
former value should be at least as large as the latter.  Suppose
$D(i,j)$ differed from  $D(i+1,j)$ by more than one.  Then we can
remove any matching for $i$.  This has at most one fewer base pair matching
and is a valid solution for the
subsequence $(i+1, j)$ with a larger value than its current value,
contradicting the optimality of $D(i+1,j)$.  An analogous argument
holds along the columns. 
\end{proof}

Once the cells $D(i, l),\, D(i, l+1),\, \ldots,\, D(i, l+q-1)$ are
computed, for some $l \in \{i,\ldots, j-q\}$, they can be represented by $D(i,l) + V_0,\, D(i, l)+V_1,\,
\ldots,\, D(i, l) + V_{q-1}$, where $V_p = D(i, l+p) - D(i, l)$, for
$p \in \{0, \ldots, q-1 \}$.  Let us define, $v_0 = 0$ and $v_p = V_p -
V_{p-1}$, for $p \in \{1, \ldots, q-1 \}$.  From lemma~\ref{lem:diff},
$v_p \in \{0,1\}$, for all $p \in [0,q-1]$.  Let $\mathbf{v}$ denote
the binary vector $v_0, v_1, \ldots, v_{q-1}$ of differences and let $\mathbf{V}$
denote the vector of running totals $V_0, V_1, \ldots, V_{q-1}$.

Since the $v_p$'s are defined from $V_p$'s, the inverse function is
well defined:  $V_p = \sum_{k=0}^{i} v_k.$  Thus $D(i, l)$ together with
the vector $\mathbf{v}$ represents $q$ consecutive cells of the table.

Similarly, since the values are non-increasing
down a column, $D(i+l+1, j), \ldots, D(i+l+q, j)$ be represented by the
pair $D(i+l+1, j), \mathbf{\bar{v}}$, where   $ \mathbf{\bar{v}} \in \{0, -1\}^q$.
We call $\mathbf{v}$ the {\em horizontal difference vector} or the
{\em horizontal vector} and we call $\mathbf{\bar{v}}$ the {\em
vertical difference vector} or the {\em column vector}.  The
corresponding vector of sums is denoted $\bar{V}$.

Consider $q$ consecutive cells from $l+1$ to $l+q$ used in computing $D(i,j)$:
\begin{eqnarray}
\label{eqn:nus} D(i,j) & \leftarrow & \max_{l+1 \leq k \leq l+q} \quad D(i, k-1) + D(k,j)\\
       & \leftarrow & \max_{0  \leq k \leq q-1} \qquad D(i,l) + V_k + D(i+l+1, j) +
       \bar{V}_k \nonumber\\
\label{eqn:4r}       & \leftarrow & D(i,l) + D(i+l+1, j) + \max_{0  \leq k \leq q-1}  \quad V_k +
       \bar{V}_k
\end{eqnarray}

As before, we use $\argmax$ in place of $\max$ to obtain $D^*(i,j),$ which facilitates the traceback.

As noted above the second line of the recurrence~(\ref{eqn:recur}),
looping over elements, is more expensive and we
will use (\ref{eqn:4r}) instead of (\ref{eqn:nus}) to compute the $D$ and
$D^*$ values in the Four-Russians method. That is, we will use
(\ref{eqn:4r}) for groups of $q$ cells each instead of one loop of
(\ref{eqn:recur}). Since the $V$ vectors are in
bijection with the $\mathbf{v}$ vectors, we will do the preprocessing
using $\mathbf{v}$.  Let $\mathbf{v}$ and  $\mathbf{\bar{v}}$ be
the corresponding vectors in (\ref{eqn:4r}).  The following algorithm
evaluates the $\max$ computation.

\algsetup{indent=2em}
\begin{algorithmic}[1]
\renewcommand{\algorithmicrequire}{\textbf{Input:}}
\REQUIRE {horizontal difference vector $\mathbf{v}$ and vertical
difference vector $\mathbf{\bar{v}}$}
\STATE $max\text{-}val \gets 0$ and $max\text{-}index \gets 0$
\STATE $sum_1 \gets 0$ and $sum_2  \gets 0$
\FOR{$k = 0$ to $q-1$}
  \STATE {$sum_1 \gets sum_1 + v_i$}
  \STATE {$sum_2 \gets sum_2 + \bar{v}_i$}
  \IF{$sum_1 + sum_2 > max\text{-}val$}
    \STATE {$max\text{-}val \gets sum_1 + sum_2 $}
    \STATE {$max\text{-}index \gets k$}
  \ENDIF
\ENDFOR
\RETURN ($max\text{-}val, max\text{-}index$)
\label{alg:preproc}
\end{algorithmic}

Using this instead of (\ref{eqn:nus}) is not advantageous in itself.
However, if this algorithm is given as a black box, $D(i,j)$ can
be computed in constant time by invoking the black box once.  In the
preprocessing stage, we will run the above algorithm for all possible
vector pairs of length $q$ and store the results in table $R$.
Table $R$ is indexed by a pair of numbers in the range $[2^q]$ to represent
the two vectors $(\mathbf{v}, \mathbf{\bar{v}})$.
Since there are two entries in the table, the lookup is a constant
time operation.
We will show later that this exhaustive enumeration is not too expensive.

%

In the Nussinov algorithm described in the previous section, the
recurrence is evaluated using (\ref{eqn:nus}) and it takes $O(q)$
time. In the Four-Russians method, using the preprocessing step, the
$\max$ computation is available through a table lookup and the recurrence
for $q$ terms can be completed in constant time.  This reduction in the
computation time is the reason for the speedup by a factor of $q$.


\begin{figure}
\centering \includegraphics[scale=0.35]{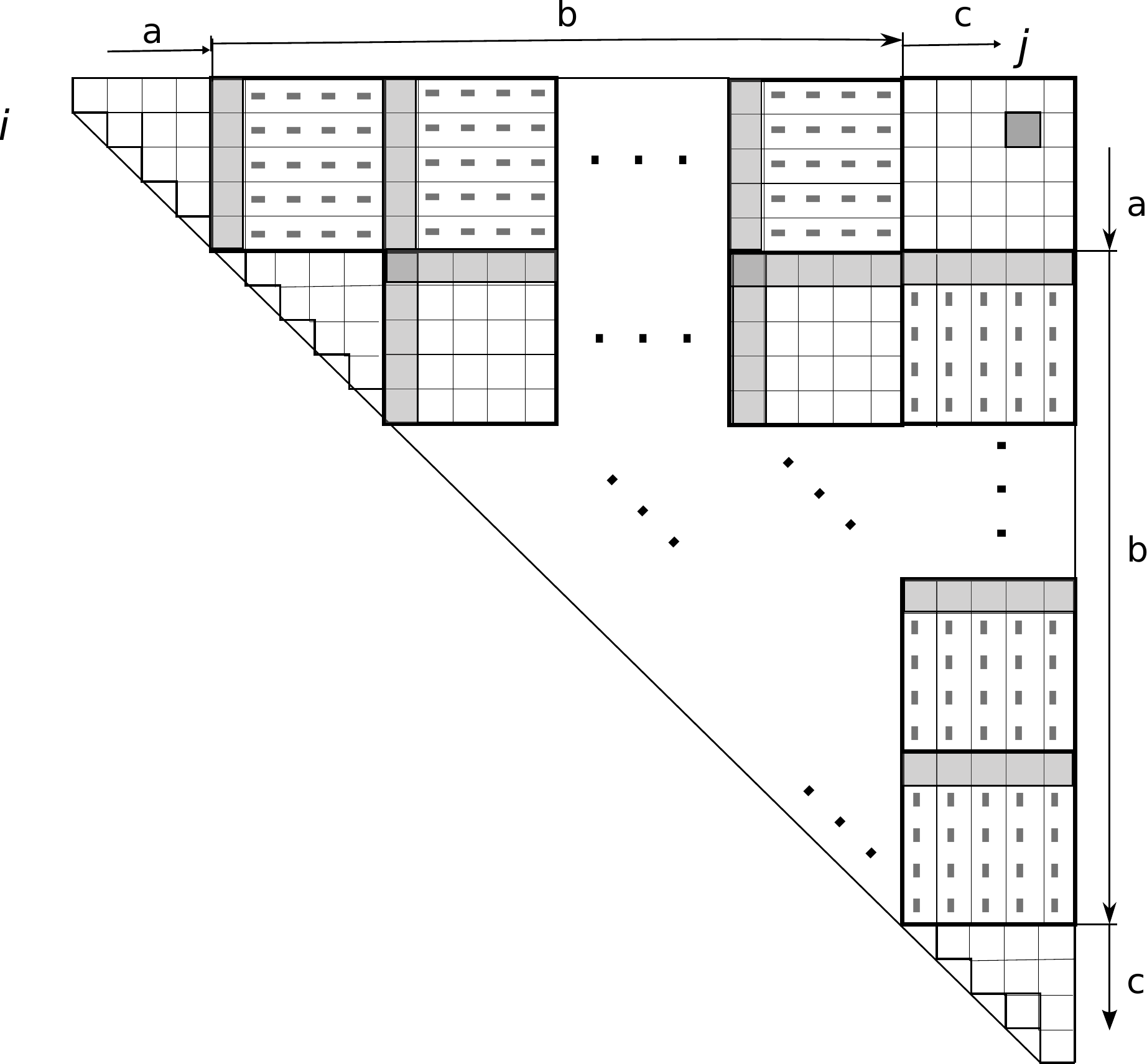}
\caption[The two-vector method.]{A diagrammatic representation of the
two-vector method. The row and column blocks are matched as labelled. The gray boxes and the gray dashes
show the initial value and difference vectors.  The group of cells in
$b$ correspond to the Four-Russians
loop in lines \ref{alg:four-rus-loop}--\ref{alg:end-four-rus-loop} of
Algorithm~\ref{alg};
the cells in $a$ are used in the loop in lines \ref{alg:first-loop}--\ref{alg:end-first-loop}
and the cells in $c$ form the loop in lines \ref{alg:sec-loop}--\ref{alg:end-sec-loop}.}
\label{fig:four-rus}
\end{figure}

The two-vector method modifies the Nussinov algorithm as
follows.  All the rows and columns of the table are grouped into groups
of $q$ cells each. The recurrence over these $q$ cells is computed in
constant time using the preprocessing table.  The recurrence involves
$D(i, k-1) + D(k, j)$, i.e., the value in the $(k-1)^\textrm{st}$
column is used with the $k^\textrm{th}$ row. Therefore the row and
column groupings differ by one.  That is, the columns are grouped $(0, 1,
\ldots, q-1)$, $(q, q+1, \ldots, 2q-1)$ etc.  The rows are grouped $(1,
2, \ldots, q)$, $(q+1, q+2, \ldots, 2q)$ etc.    This ensures that the
row and column groups are well characterized.  That is, to fill the
cell $(i,j)$, the $k^{\textrm{th}}$ group along row $i$ needs to be
combined with the $k^{\textrm{th}}$ group below $(i,j)$ in column $j$.

The cells of the table are filled in the same order as before.  When the
last cell of a row- or a column- group is evaluated the corresponding row
and column vectors are computed and stored.  To fill cell $(i,j)$,
we retrieve the first element and the horizontal vector
of the group from row $i$ and the first element and the column vector from the
corresponding group in column $j$.  The recurrence is solved using
(\ref{eqn:4r}) by a table lookup.  The final value for $D(i,j)$ is the
maximum value over all the groups. There might be residual elements in
the row that do not fall in these groups.  There are at most $2q$ such
elements. These are solved separately
using Nussinov's method. Algorithm~\ref{alg} has the algorithm listing and
Figure~\ref{fig:four-rus} describes the algorithm pictorially.

\renewcommand{\algorithmiccomment}[1]{$\qquad$ // #1}
\begin{algorithm}[t!]
\caption{Procedure for the two-vector Four-Russians speedup. The DP table is filled column-wise.}
\label{alg}
\begin{algorithmic}[1]
\STATE {$R \gets $ preprocess all pairs of vectors of length $q$}
\FOR{$j = 1$ to $n$ }
  \STATE $D(j,j) \gets 0$
  \FOR{$i = j-1$ down to $1$ }
    \label{alg:ij} \STATE $D(i,j) \gets b(S[i], S[j]) + D(i+1, j-1)$
    \STATE Let $(i,i)$ be in the $I^{\textrm{th}}$ group  in row $i$.
    \STATE Let $(i,j)$ be in the $J^{\textrm{th}}$ group horizontally
    in the $i^{\textrm{th}}$ row and $J'^{\textrm{  th}}$ group
    vertically in the $j^{\textrm{th}}$ column.
    \STATE {Let $i_q$ be the right-most entry of group $I$ and $j_q$ be the left-most entry 
    in group $J$}
    \FOR[For all cells in the first group] {$k = i+1$ to $i_q $} \label{alg:first-loop}
      \STATE $D(i,j) \gets \max(D(i,j), \,D(i, k-1)+D(k, j))$
    \ENDFOR \label{alg:end-first-loop}
    \FOR[For all cells in the last group]{$k = j_q$ to $j$}\label{alg:sec-loop} 
      \STATE $D(i,j) \gets \max(D(i,j), \,D(i, k-1)+D(k, j))$
    \ENDFOR \label{alg:end-sec-loop}
    \FOR[For all groups in between]{$k = 1$ to $J-I$ } \label{alg:four-rus-loop}
      \STATE Let $p$ be the left-most cell in the $k^{\textrm{th}}$
      group to the right of $I$ and $q$ be the top-most cell in
      the $k^{\textrm{th}}$ group below $J'$.
      \STATE Let $v_p$ and $v_q$ be the corresponding horizontal and
      vertical difference vectors.
      \STATE $D(i,j) \gets \max( D(i,j), \, D(i,p) + D(q,j) + R(v_p,
      v_q))$ \label{alg:four-rus-step}
    \ENDFOR \label{alg:end-four-rus-loop} 
    \IF[compute the vertical difference vector]{$i \mod q = 0$} 
      \STATE compute and store the $\mathbf{v}$ vector
      $i/q^\textrm{th}$ group for column $j$
    \ENDIF
    \IF[compute the horizontal difference vector]{$j \mod q = q-1$} 
      \STATE compute and store the $\mathbf{v}$ vector
      $(j-1)/q^\textrm{th}$
      group for row $i$
    \ENDIF
  \ENDFOR  

\ENDFOR 
\end{algorithmic}
\end{algorithm}
\paragraph{Runtime Analysis.}  In the precomputation phase,
there are $2^q$ $q$-length vectors and $2^{2q}$ pairs of vectors.
The precomputation takes $O(q)$ time per vector pair.  Thus the total
time for precomputation is $O(q2^{2q}).$

The main algorithm:  There are $O(n^2)$ cells and to fill each cell it
takes $O(n/q + q)$ time.  That is,  it takes $O(n/q)$ time to look up
the initial value and the difference vector and the $R$ table lookups
for the the $O(n/q)$
groups. It takes $O(q)$ time for the residual
elements. Thus it takes $O(n^2 \times (n/q + q))$ time to fill the
table. Every cell is involved in at most two vector computations, where
the difference to its neighbor is computed once for the row and for the
column vector. This takes an amortized $O(n^2)$ time which is dominated
by the rest of the algorithm.

When $q = \log n$, the total time for the entire algorithm is $O(\log
n\, 2^{2\log n}\, + n^2 \, + \,n^2\times (\frac{n}{\log n} + \log n))
= O(n^2 \log n + n^3/\log n) = O(n^3/log n)$.

\subsection{Other Variants}
\paragraph{FG Algorithm.}

Frid and Gusfield~\cite{FG} first showed how the Four-Russians approach
could be applied to the RNA-folding problem.  We will call their algorithm
the {\em FG} algorithm.  {\em FG} and {\em two-vector} algorithms are
variants of the same idea.  We will highlight the differences in
preprocessing and the maximum value computation by the
Four-Russians technique. In particular, we will show the maximum
computation in step~\ref{alg:four-rus-step} of Algorithm~\ref{alg}.

After computing the $q$-contiguous cells of a group in a row, the
value in the initial cell $D(i,p)$ and the horizontal
difference vector $v_p$ are known. They run the preprocessing
algorithm in page~\pageref{alg:preproc} for this
fixed $v_p$ vector together with all possible vertical difference
vectors.  They add the value of $D(i,p)$ to the maximum and
table the result. This preprocessing step is computed for every block
of every row. 
The preprocessing table $R$ is indexed by row number, group number and
a vector (which is a potential column vector).  The horizontal vectors
need not be stored. 

To fill cell $(i,j)$, they iterate over all groups and find the $q$-length
column vectors. The preprocessed value for this vector in the
corresponding block is retrieved from the table and the result is
added to $D(q,j)$.

%

The preprocessing is for horizontal vectors seen in the table.
Since the horizontal vectors are not known beforehand, the precomputation
cannot be done prior to the main algorithm.  Instead, it is interleaved
with the computation of the table.  They fill part of the DP table and use
the vectors to complete some preprocessing,  which in turn is used
fill another part of the table and so on.

Since the preprocessing is done for every group of every row, the same
horizontal vector can be seen multiple times in the table.  This leads
to duplicated work and slower running time than the two-vector algorithm.

\paragraph{Memoization} The two-vector method computes the
preprocessing over all possible vector pairs and the FG method
for only the horizontal vectors that are seen in the table.  Stated
this way, a hybrid approach suggests itself.

In our next variants, we memoize the results for a pair of vectors. Like
the two-vector approach, the preprocessing is done only once for a
vector pair and like the FG algorithm, it is only for the vectors
seen in the table and the preprocessing is interleaved with the main
algorithm. Since the preprocessing table is indexed by two vectors,
unlike the FG algorithm, the results are computed only once for every
vector seen.

In the partially memoized version, upon completion of elements of a
group,
if a new horizontal vector is seen, we pair it with all possible $2^q$
column vectors and the results are tabled.  In the completely memoized
version, the result for a pair of vectors is computed the first time
the pair is observed and the result is stored in the table.  The result for
future occurrences of the same pair are obtained by a table lookup.
the rest of the algorithm is identical to the two-vector method.

All these variants take $O(n^3/\log n)$ time but the memoized versions
potentially store fewer vectors than the two vector method and will
have a similar worst-case runtime in practice as the two-vector method.  But, as argued before,
the FG method does duplicated work and will be slower in practice.

\section{Parallel Algorithm}
\label{sec:par}
The Nussinov DP algorithm can be parallelized with $n$ processes to get an
$O(n^2)$ parallel algorithm.  We assign one parallel process to a
column.  In the $i^\textrm{th}$ iteration, each process computes the value for the
$i^\textrm{th}$ diagonal entry.  That is, the successive diagonals are 
solved in iterations and in each iteration the entries of the diagonal
are solved in parallel.  To compute the value for cell $(i,j)$, the
entries in the row to its left and in the column below $(i,j)$ are
needed.  Since these values are computed in earlier iterations, each
diagonal cell can be filled independent of the other processes.

A process has to compute the value for $O(n)$ cells and for each cell
it needs to access $O(n)$ other cells. Thus the total computation
takes $O(n^2)$ time with $n$ processes.

The parallel algorithm for process $j$ for $j = 1, 2, \ldots, n$:
\algsetup{indent=2em}
\begin{algorithmic}[1]
\STATE {$D(j+1,j) \gets 0$, $D(j,j) \gets 0$}
\FOR {$i = j$ down to 1 } 
  \STATE $D(i,j) \gets D(i+1, j-1) + b(S[i], S[j])$
  \label{par:loop} \FOR {$k= i+1$ to $j$}
    \STATE {$D(i,j) \gets \max\{ D(i,j), D(i, k-1) + D(k,j) \} $}
  \label{par:end-loop} \ENDFOR
  \STATE{ Synchronize with other processes}
\ENDFOR
\end{algorithmic}

We will describe the use the two-vector Four-Russians method to obtain an
$O(n^2/\log n)$ algorithm below. The preprocessing step that enumerates
the solution for $2^{q} \times 2^{q}$ difference vectors is embarrassingly
parallel and we do not discuss the parallel algorithm for it.

As before, we have $n$ processes one for each column.  Each process solves
the entries of the column from bottom to top.  Instead of computing
the maximum over each cell in the inner loop (lines~\ref{par:loop} --
\ref{par:end-loop} in the parallel algorithm above), we use the Four-Russians
technique to  solve $q$ cells in one step by looking up the table computed
in the preprocessing step.

Let $d_H(i,j)$ be the horizontal difference vector for cells $D(i,j),
\ldots, D(i+q-1, j)$  and let $d_V(i,j)$ be the vertical difference
for cells $D(i,j), \ldots, D(i+q-1, j)$.  We modify the inner loop of
the parallel algorithm as follows:

\algsetup{indent=2em}
\begin{algorithmic}[1]
\FOR {$k'= 0$ to  $\lfloor j/q \rfloor -1$}
\STATE {$k = i + k'*q$}
\STATE{ $D(i,j) = \max\{ D(i,j), 
D(i, k) + D(k+1,j) + R[d_H(i,
k)][d_V(k+1, j)] \} $}
\ENDFOR
\FOR {$ k = \lfloor j/q \rfloor \times q$ to $j$}
\STATE {$D(i,j) \gets \max\{ D(i,j), D(i, k) + D(k+1,j) \} $}
\ENDFOR
\STATE {Compute the horizontal and vertical differences and store them
in $d_H(i-q+1,j)$ and $d_V(i,j)$ respectively.}
\end{algorithmic}

For each entry, the first loop takes $O(n/q)$ time and the second loop
takes $O(q)$ time.  Since all the processes are solving the
$k^\textrm{th}$ diagonal in the $k^\textrm{th}$ iteration, all of them
execute the same number of steps before synchronization.  Note that we
compute the horizontal and vertical differences for every node, unlike
in section~\ref{sec:two-vec} where they are computed every
$q^\textrm{th}$ cell, to ensure that every process performs the same
number of steps and simplify the analysis.
The difference vectors can be computed in $O(q)$
time.   These can also be computed in constant time by shifting the
previous difference vector and appending the new difference.  But we
will not assume this simplification for the time bound computation.

Thus each entry can be computed in $O(n/q + q)$ time.  There are $O(n)$
entries for each process, thus the total time taken for all processes
to terminate is $O(n^2/q + nq)$.  With $q = \log n$ as before, this
gives an $O(n^2/ \log n)$ algorithm.  

\section{Parallel Implementation}
\label{sec:CUDA}
\subsection{GPU Architecture}
Graphics processing units (GPUs) are specialized processors designed
for computationally intensive real-time graphics rendering.
Compute Unified Device Architecture (CUDA)
is the computing engine designed by NVIDIA for their GPUs.

The programmer can group threads in a {\em block}, which in turn can be
organized in a {\em grid} hierarchy. Memory hierarchy includes thread-specific
local memory, block-level shared memory for all threads in the block
and global memory for the entire grid.  The access times increases along
the hierarchy from local to global memory.

Since the access to global memory is
slower (more clock cycles than local memory access), it is efficient for
the  threads within a block to access contiguous memory locations. Then
the hardware {\em coalesces} memory accesses for all threads in a
block into one request.
More specifically, in our application, if a matrix is stored in row-major
order and if the threads in a block access contiguous elements of a row,
then the accesses can be coalesced.  However, accessing elements along
a column is inefficient as distant memory elements have to be fetched
from different cache lines.

Programs that observe the hardware specifications can exploit the
optimizations in the system and are fast in practice.  We designed
the program that exploits the parallel structure of the DP
algorithm and the hardware features of the GPU.  


\subsection{Related Work}
\label{sec:cuda-rel} 
As mentioned earlier, the cells of a diagonal are independent of
one another and can be computed in parallel.  In Stojanovski et
al.~\cite{SGM2011}, elements of the diagonal are assigned to a block
of threads.  This design does not handle memory coalescence for either
row or column accesses.  Chang et al.~\cite{Chang2010} allocate an $n
\times n$ table and reflect the upper-triangular part of the matrix on
the main diagonal.  Successive elements of a column are fetched from
the row in the reflected part of the matrix. When threads of a block are
assigned to elements of a diagonal, the successive column accesses for
a thread are to consecutive memory cells.  However, this does not allow
coalesced access for threads within a block. Rizk and Lavenier~\cite{rizk}
show an implementation for RNA folding under energy models.  They
show a tiling scheme where a group of cells are assigned to a block of
threads to reuse the data values that are fetched from a
column. In this paper, we show that storing the row and column vectors in
different orders for two-vector method can further improve the efficiency.

\subsection{Design of the Four-Russians CUDA Program}
We briefly describe the design of the CUDA program; 
a longer discussion can be found in~\cite{RNA-tr}.

We group 
cells together into tiles,   
where each tile is a composite of $q \times q$ cells.  The tiles along
a diagonal can be computed independent of each other. Each tile is
assigned to a block of threads and computed in parallel.  After all the entries of the tile are
computed, only the horizontal and vertical differences are stored.



To fill a tile, the horizontal differences of all the tiles to the
left and vertical differences from the tiles underneath are accessed.
These difference vectors are stored in different orders.  The
horizontal difference vectors of the rows of a tile are stored in
contiguous memory locations
and the tiles are stored in row-major order.  The vertical difference
vectors of the columns of a tile are stored together and
the tiles are grouped in column-major order.


To fill a tile, the horizontal difference vectors from a tile to the left
are fetched.  When each thread retrieves one vector, the block of
threads accesses contiguous memory locations and the memory accesses
are
coalesced.  Successive iterations fetch the tiles along a row which
are in contiguous memory locations.  Similarly the vertical differences of
a tile below are accessed in one coalesced memory access by the
threads of the block.  

\section{Empirical Results}
\label{sec:rna-expts}

Prior to empirical evaluation, the FG algorithm was expected to be the
slowest due to redundant work.  The memoized versions were expected to
be faster than the two-vector algorithm, as they preprocess only a
subset of the $2^{2q}$ vectors seen in the table. 

We ran the programs on complete mouse non-coding RNA sequences. We
also tested the performance on random substrings on real RNA sequences
and random strings over {\tt A,C,G,U}.

The FG algorithm, while faster than Nussinov, was the slowest among the
Four-Russians methods, as expected.  The completely memoized version
was slower than the other two variants.  This is because every lookup
of the preprocessing table includes a check to see if the pair of
vectors has already been processed.  There are $2^{2q}$ unique vector
pairs but there are $O(\frac{n^3}{q})$ queries to the preprocessing
table and each query involves checking if the vector pair has been
processed plus the processing time for new pairs.  There are
$O(\frac{n^2}{q^2})$ vector pairs in the table. For larger $n$
(eg., $n > 1000$ and $q=8$), 
all the $2^{2q}$ vectors are expected to be present in
the DP table. Generally, memoized subproblems are relatively expensive
compared to the lookup.  Since the preprocessing here has only $q$
steps, the advantage of memoization is not seen.

The partially memoized version was slightly slower than the two vector
algorithm.  Again, the advantage of potentially less preprocessing
than the two-vector method  is erased by the need
to check if a vector has been processed. The two-vector
method was the fastest on all sequence lengths tested.

For short sequences the two vector method took negligible time (less than 0.2
seconds up to 1000 bases) and are not reported.  For longer sequences,
we noticed that using longer vector lengths reduced the running time
and the improvement saturated beyond $q=8$ or 9.  Beyond this, the
extra work in preprocessing overshadowed the benefit.
A similar trend was seen for the memoized versions
too.  However, for the FG method $q=3$ gave the best speedup and
longer vector lengths had a slower running time due to the extra
preprocessing at every group.

All the programs were written in C++ compiled with the highest compiler optimizations.
We only discuss the experimental results on a desktop and two GPU cards
in this paper.  Detailed notes on running times can
be found in~\cite{RNA-tr}.

We measured the running times of the different versions of our
serial algorithms on a desktop machine with a Pentium II 3GhZ processor and 1MB cache.
The running times of Nussinov and the speedups of
various programs compared to Nussinov are shown in the table below.
For sequences of length 6000, the two-vector method takes close to a minute on the desktop.
\begin{center}
\begin{tabular}{|c||c|c|c|c|c|}
\multicolumn{6}{c}{Speedup factors of the serial programs on the desktop} \\
\hline
& Time & \multicolumn{4}{c|} {Speedup} \\
\hline
Length & Nussinov & Two-vector & Partially & Completely & FG \\
&(in secs) & & Memoized & Memoized &  \\
\hline
2000 & 16.5  & 7.7  & 7.3  & 5.6  & 3.0  \\
\hline 
3000 & 62.5  & 8.8  & 8.3  & 6.4  & 3.4  \\
\hline 
4000 & 196.6 & 11.9 & 11.4 & 8.8  & 4.7  \\
\hline 
5000 & 630.3 & 21.1 & 18.9 & 14.7 & 7.8  \\
\hline 
6150 & 1027.8 & 18.1 & 17.0 & 13.3 & 7.03\\
\hline
\end{tabular}
\end{center}
\begin{figure}[t!]
\begin{center}
\begin{subfigure}{0.5\textwidth}
\hspace{-10pt}
\includegraphics[scale=0.4]{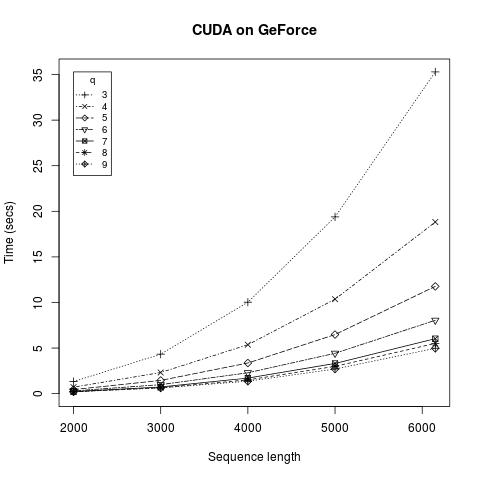}
\end{subfigure}
\quad
\begin{subfigure}{0.45\textwidth}
\centering \includegraphics[scale=0.4]{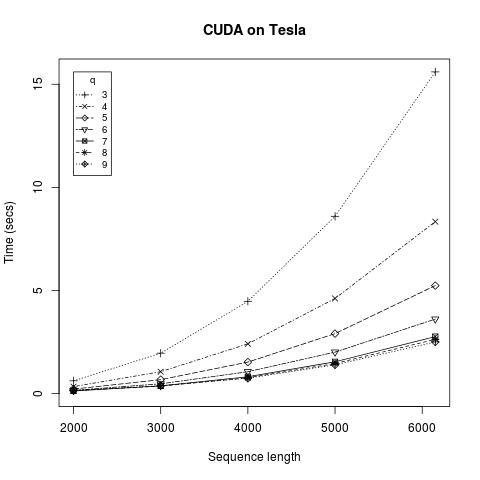}
\end{subfigure}
\end{center}
\caption[Running time of the CUDA program on two GPUs.]{Running time of the CUDA program on two GPUs. The programs
run twice as fast on the Tesla card than the GeForce card.}
\label{fig:cuda}
\end{figure}

Fig.~\ref{fig:cuda} shows the execution times on two GPU cards --
GeForce GTX 550 Ti card with 1GB on-card memory and 
Tesla C2070 with 5GB memory. The programs take about a second
for sequences up to 4000 bases long, and takes about 5 seconds and 2.5
seconds for sequences of length 6000.  The running times for various
sequence lengths are shown in the table below.
  \begin{center}
  \begin{tabular}{|c||c|c|}
  \multicolumn{3}{c} {Running times for the parallel program (in
  secs)}\\
  \hline
  Length & On GeForce & On Tesla \\
  \hline
  2000 & 0.20 & 0.14 \\
  \hline
  3000 & 0.62 & 0.38 \\
  \hline
  4000 & 1.36 & 0.74 \\
  \hline
  5000 & 2.70 & 1.39 \\
  \hline
  6000 & 4.97 & 2.50 \\
  \hline
  \end{tabular}
  \end{center}

\section{Conclusions and Future Work}

We described the two-vector method for using the Four-Russians technique
for RNA folding.  This method is simpler than the Frid-Gusfield
method.  It also improves the bound of the parallel algorithm by a
$\log n$ factor to $O(\frac{n^2}{\log n})$.
We showed two other variants that memoize the preprocessing results.
These methods are faster than Nussinov by up to a factor of 20 and the
Frid-Gusfield method by a factor of 3.

In the future, it will be interesting to see the application of the
Four-Russians technique for other methods that use energy models with
thermodynamic parameters.  The Frid-Gusfield method has been applied
to RNA co-folding~\cite{FG-cofolding} and folding with
pseudoknots~\cite{FG-psuedo} problems; the application of the two-vector method
to those problems and its implications are also of interest.  It will be
interesting to compare our run time with the other improvements over
Nussinov, like the boolean matrix multiplication method~\cite{Akutsu1999}.

\section*{Acknowledgements}
The first-listed author thanks Prof.~Norm Matloff for the opportunity
to lecture in his class; this project spawned out of that lecture.
Thanks also to Prof.~John Owens for access to the server with a Tesla
card.  Thanks to Jim Moersfelder and Vann Teves from Systems Support Staff
for help in setting up the CUDA systems.

%% file: wabi.bbl
\begin{thebibliography}{10}

\bibitem{Akutsu1999}
T.~Akutsu.
\newblock Approximation and exact algorithms for {RNA} secondary structure
  prediction and recognition of stochastic context-free languages.
\newblock {\em J. Comb. Optim.}, 3(2-3):321--336, 1999.

\bibitem{Andronescu2007}
M.~Andronescu, A.~Condon, H.~H. Hoos, D.~H. Mathews, and K.~P. Murphy.
\newblock Efficient parameter estimation for {RNA} secondary structure
  prediction.
\newblock {\em Bioinformatics}, 23(13):i19--i28, 2007.

\bibitem{Andronescu2010}
M.~Andronescu, A.~Condon, H.~H. Hoos, D.~H. Mathews, and K.~P. Murphy.
\newblock Computational approaches for {RNA} energy parameter estimation.
\newblock {\em RNA}, 16(12):2304--2318, 2010.

\bibitem{four-rus}
V.~Arlazarov, E.~Dinic, M.~Kronrod, and I.~Faradzev.
\newblock On economical construction of the transitive closure of a directed
  graph (in {R}ussian).
\newblock {\em Dokl. Akad. Nauk.}, 194(11), 1970.

\bibitem{Backofen2011Zakov}
R.~Backofen, D.~Tsur, S.~Zakov, and M.~Ziv-Ukelson.
\newblock Sparse {RNA} folding: Time and space efficient algorithms.
\newblock {\em Journal of Discrete Algorithms}, 9(1):12 -- 31, 2011.

\bibitem{Chang2010}
D.-J. Chang, C.~Kimmer, and M.~Ouyang.
\newblock Accelerating the nussinov {RNA} folding algorithm with {CUDA/GPU}.
\newblock In {\em {ISSPIT}}, pages 120--125. IEEE, 2010.

\bibitem{Ding2003}
Y.~Ding and C.~E. Lawrence.
\newblock A statistical sampling algorithm for {RNA} secondary structure
  prediction.
\newblock {\em Nucleic Acids Research}, 31(24):7280--7301, 2003.

\bibitem{Do2006}
C.~B. Do, D.~A. Woods, and S.~Batzoglou.
\newblock {CONTRA}fold: {RNA} secondary structure prediction without
  physics-based models.
\newblock {\em Bioinformatics}, 22(14):e90--e98, 2006.

\bibitem{Dowell2004}
R.~Dowell and S.~Eddy.
\newblock Evaluation of several lightweight stochastic context-free grammars
  for {RNA} secondary structure prediction.
\newblock {\em BMC Bioinformatics}, 5(1):71, 2004.

\bibitem{Durbin1998}
R.~Durbin, S.~R. Eddy, A.~Krogh, and G.~Mitchison.
\newblock {\em Biological Sequence Analysis: Probabilistic Models of Proteins
  and Nucleic Acids}.
\newblock Cambridge University Press, 1998.

\bibitem{FG}
Y.~Frid and D.~Gusfield.
\newblock A simple, practical and complete ${O}(n^3)$-time algorithm for {RNA}
  folding using the {F}our-{R}ussians speedup.
\newblock {\em Algorithms for Molecular Biology}, 5:13, 2010.

\bibitem{FG-cofolding}
Y.~Frid and D.~Gusfield.
\newblock A worst-case and practical speedup for the {RNA} co-folding problem
  using the {\it {f}our-{r}ussians} idea.
\newblock In {\em WABI}, pages 1--12, 2010.

\bibitem{FG-psuedo}
Y.~Frid and D.~Gusfield.
\newblock Speedup of {RNA} pseudoknotted secondary structure recurrence
  computation with the {F}our-{R}ussians method.
\newblock In {\em COCOA}, pages 176--187, 2012.

\bibitem{Hamada2009}
M.~Hamada, H.~Kiryu, K.~Sato, T.~Mituyama, and K.~Asai.
\newblock Prediction of {RNA} secondary structure using generalized centroid
  estimators.
\newblock {\em Bioinformatics}, 25(4):465--473, 2009.

\bibitem{Hofacker2003}
I.~L. Hofacker.
\newblock Vienna {RNA} secondary structure server.
\newblock {\em Nucleic Acids Research}, 31(13):3429--3431, 2003.

\bibitem{Knudsen2003}
B.~Knudsen and J.~Hein.
\newblock Pfold: {RNA} secondary structure prediction using stochastic
  context-free grammars.
\newblock {\em Nucleic Acids Research}, 31(13):3423--3428, 2003.

\bibitem{Lu2009}
Z.~J. Lu, J.~W. Gloor, and D.~H. Mathews.
\newblock Improved {RNA} secondary structure prediction by maximizing expected
  pair accuracy.
\newblock {\em RNA}, 15(10):1805--1813, 2009.

\bibitem{Markham2008}
N.~R. Markham and M.~Zuker.
\newblock {UNAF}old.
\newblock {\em Bioinformatics}, 453:3--31, 2008.

\bibitem{Mathews2004b}
D.~H. Mathews, M.~D. Disney, J.~L. Childs, S.~J. Schroeder, M.~Zuker, and D.~H.
  Turner.
\newblock Incorporating chemical modification constraints into a dynamic
  programming algorithm for prediction of {RNA} secondary structure.
\newblock {\em {PNAS}}, 101(19):7287--7292, 2004.

\bibitem{Nussinov1980}
R.~Nussinov and A.~B. Jacobson.
\newblock Fast algorithm for predicting the secondary structure of
  single-stranded {RNA}.
\newblock {\em {PNAS}}, 77(11):6309--6313, 1980.

\bibitem{Nussinov1978}
R.~Nussinov, G.~Pieczenik, J.~R. Griggs, and D.~J. Kleitman.
\newblock Algorithms for loop matchings.
\newblock {\em SIAM Journal on Applied Mathematics}, 35(1):68--82, 1978.

\bibitem{Reuter2010}
J.~Reuter and D.~Mathews.
\newblock {RNA}structure: software for {RNA} secondary structure prediction and
  analysis.
\newblock {\em BMC Bioinformatics}, 11(1):129, 2010.

\bibitem{rizk}
G.~Rizk and D.~Lavenier.
\newblock {GPU} accelerated {RNA} folding algorithm.
\newblock In {\em International Conference on Computational Science}, 2009.

\bibitem{SGM2011}
M.~Stojanovski, D.~Gjorgjevikj, and G.~Madjarov.
\newblock Parallelization of dynamic programming in nussinov {RNA} folding
  algorithm on the {CUDA} {GPU}.
\newblock In {\em ICT Innovations}, 2011.

\bibitem{Tinoco1973}
I.~Tinoco et~al.
\newblock {Improved Estimation Of Secondary Structure In Ribonucleic-Acids}.
\newblock {\em Nature-New Biology}, 246(150):40--41, 1973.

\bibitem{RNA-tr}
B.~Venkatachalam, Y.~Frid, and D.~Gusfield.
\newblock Faster algorithms for {RNA}-folding using the {F}our-{R}ussians
  method.
\newblock UC Davis Technical report, 2013.

\bibitem{Wexler2007}
Y.~Wexler, C.~B.-Z. Zilberstein, and M.~Ziv-Ukelson.
\newblock A study of accessible motifs and {RNA} folding complexity.
\newblock {\em Journal of Computational Biology}, 14(6):856--872, 2007.

\bibitem{Zakov2011}
S.~Zakov, Y.~Goldberg, M.~Elhadad, and M.~Ziv-Ukelson.
\newblock Rich parameterization improves {RNA} structure prediction.
\newblock In V.~Bafna and S.~C. Sahinalp, editors, {\em Research in
  Computational Molecular Biology ({RECOMB})}, volume 6577, pages 546--562.
  Lecture Notes in Computer Science, Springer, 2011.

\bibitem{Zakov2010}
S.~Zakov, D.~Tsur, and M.~Ziv-Ukelson.
\newblock Reducing the worst case running times of a family of {RNA} and {CFG}
  problems, using {V}aliant's approach.
\newblock In {\em {WABI}}, pages 65--77, 2010.

\bibitem{Zuker2003}
M.~Zuker.
\newblock Mfold web server for nucleic acid folding and hybridization
  prediction.
\newblock {\em Nucleic Acids Research}, 31(13):3406--3415, 2003.

\bibitem{Zuker1981}
M.~Zuker and P.~Stiegler.
\newblock Optimal computer folding of large {RNA} sequences using
  thermodynamics and auxiliary information.
\newblock {\em Nucleic Acids Research}, 9(1):133--148, 1981.

\end{thebibliography}
